\documentclass{article}
\usepackage[margin=1in]{geometry}
\usepackage{amsmath, amssymb, amsthm}
\usepackage{hyperref}
\hypersetup{
    colorlinks=true,
    linkcolor=blue,
    citecolor=blue,
    urlcolor=blue,
}
\usepackage[square,numbers]{natbib}
\usepackage[caption=false,font=footnotesize]{subfig}
\usepackage{graphicx}
\usepackage{xcolor}
\usepackage{authblk}
\newtheorem{theorem}{Theorem}
\newtheorem{assum}{Assumption}
\newtheorem{lemma}{Lemma}

\newcommand{\SO}[1]{\ensuremath{\textrm{SO}(#1)} }
\newcommand{\SE}[1]{\ensuremath{\textrm{SE}(#1)} }
\newcommand{\so}[1]{\ensuremath{\mathfrak{so}(#1)} }

\newcommand{\dotpr}[2]{\ensuremath{\langle #1, #2 \rangle }}

\def\Pa{\ensuremath{\mathbb{\pi}_{\so{3}} }}
\def\R{\ensuremath{\mathbb{R}} }
\def\N{\ensuremath{\mathbb{N}} }

\graphicspath{{./Figures/}}
\title{Angular velocity and linear acceleration measurement bias estimators for the rigid body system with global exponential convergence}
\date{}
\author{Soham Shanbhag\thanks{sshanbhag@kaist.ac.kr} }
\author{Dong Eui Chang\thanks{Corresponding author, dechang@kaist.ac.kr}}
\affil{School of Electrical Engineering, Korea Advanced Institute of Science and Technology, Daejeon, Republic of Korea}
\begin{document}

\maketitle

\begin{abstract}
    Rigid body systems usually consider measurements of the pose of the body using onboard cameras/LiDAR systems, that of linear acceleration using an accelerometer and of angular velocity using an IMU.
    However, the measurements of the linear acceleration and angular velocity are usually biased with an unknown constant or slowly varying bias.
    We propose a measurement bias estimator for such systems under assumption of boundedness of angular velocity.
    We also provide continuous estimates to the state of the system, i.e. the pose, linear velocity, and position of the body.
    These estimates are globally exponentially convergent to the state of the rigid body system.
    We propose two bias estimators designed with the estimate of the pose in the ambient Euclidean space of the Special Euclidean group and show global exponential convergence of the proposed observers to the state of the system.
    The first observer assumes knowledge of bounds of the angular velocity, while the second observer uses a Riccati observer to overcome this limitation.
    We show the convergence with an example of a rigid body rotation and translation system on the special Euclidean group.
    We show that the observer is able to estimate the bias using data collected from an Intel Realsense camera.
\end{abstract}

\section{Introduction}%<<<
Research in quadcopter design and control has been increasing due to their use in military and civilian applications\cite{BashiHASW2017}.
For effective use of these systems, the study of control of these systems is paramount.
However, for feedback based control, one important design consideration is the integration of sensors in those systems for providing the state of the system.
Sensors proposed for this quadcopter system usually measure the position of various landmarks in the environment, given by the equation
\begin{align*}
    y_i = R^T(p_i - p),
\end{align*}
where $R$ is the matrix denoting the rotation of the quadcopter frame with respect to the ground frame, $p$ is the position of the quadcopter in the ground frame, and $p_i$s and $y_i$s are the position of the landmarks in the ground and quadcopter frame, respectively.
Moreover, measurements concerning angular velocity are collected using an IMU and linear velocity measurements are collected using GPS sensors or Doppler velocity sensors\cite{BrasISSO2016,VascoCSO2010,VascoCSO2007}.
These problems are known as ``Attitude and Heading Reference System" problems in the literature\cite{Farre2008}.
However, the measurement of linear velocity is difficult if velocity measurements are not available, like in GPS denied environments.
Alternatively, accelerometer measurements have been also exploited as an alternative to linear velocity measurements due to their ease of use.

Usually, IMU and accelerometer sensor measurements are corrupted with a constant bias.
Deterministic observers for state and bias estimation have been designed to estimate this bias in an online fashion.
The authors in \cite{Hashi2021} design an almost semi-globally uniformly ultimately bounded observer for this system.
The authors in \cite{deMaHHS2020} and \cite{HuaA2018} propose locally exponentially stable observers using the continuous Riccati equation.
The authors in \cite{BarraB2017} propose an asymptotically stable observer using invariant extended Kalman filters.
The authors in \cite{BatisSO2011} propose a globally convergent observer for the bias in the accelerometer readings, but assume unbiased angular velocity measurements.

As can be seen, the observers proposed are unable to achieve global convergence in the presence of IMU and accelerometer bias.
This is due to a topological obstruction on $\SO{3}$ not allowing global observers\cite{BhatB1998}.
Two methods are used to overcome this obstruction.
The authors in \cite{WangT2020a} design a globally exponentially stable observer for the rigid body system.
They design a hybrid (continuous + discrete) observer for the rigid body system with the state consisting of the rotational pose, translation pose, and the linear velocity.
To the best of our knowledge, this is the first observer proposing global exponential stability.
However, the hybrid nature of the observer leads to discontinuities in the observer dynamics, specifically the pose.
Since the estimate of the pose is used in classical control methods, a discontinuity in the estimate of the pose is undesirable.
Another alternative is the extension of the dynamics to the ambient Euclidean space.
This approach has been used to solve similar problems, for example in \cite{Chang2021,ParkPKC2021}.
We use this approach to design globally convergent continuous observers for the rigid body system.
An observer using a similar approach has been designed for the system considering linear velocity measurements instead of linear acceleration measurements in \cite{ShanbC2022a}.

The ambient space extension allows us to design deterministic observers which are continuous and globally convergent.
We propose two observers for this purpose.
The first observer assumes knowledge of bounds of the angular velocity and proposes a constant gain observer.
The second observer is a Riccati based variable gain observer, which is computationally expensive but assumes only the existence of a bound on the angular velocity, and no knowledge of the said bound.
To our knowledge, these are the first observers able to provide global exponential convergence without discontinuities.

%>>>

\section{Preliminaries}%<<<
We denote the estimate of the state $A$ by $\bar{A}$.
Denote by $\SO{3}$ the set of orthogonal matrices of dimension 3 with determinant 1, and by $\so{3}$ the corresponding Lie algebra.
The elements of $\so{3}$ are skew symmetric matrices in $\R^{3 \times 3}$.
The matrix representation of the cross product with a vector $v$ is denoted by $v_\times: \R^3 \to \so{3}$ such that for all $w \in \R^3, v \times w = v_\times w$.
The inverse of the $(\cdot)_\times$ operator is denoted by $(\cdot)_\vee: \so{3} \to \R^3$, where for all $a \in \R^3, (a_\times)_\vee = a$ and for all $A \in \so{3}, (A_\vee)_\times = A$.
The Euclidean inner product of two matrices in $\R^{m \times n}$ is denoted by $\dotpr{\cdot}{\cdot}: \R^{m \times n} \times \R^{m \times n} \to \R$ such that $\dotpr{A}{B} = \textrm{trace}(A^T B)$.
The Euclidean norm of a matrix $A \in \R^{m \times n}$ is defined as $\| A \| = \sqrt{\langle A, A \rangle}$.
The orthogonal projection of $A \in \R^{3 \times 3}$ to $\so{3}$ is given by $\Pa(A) = (A - A^T)/2$.
Let $A \otimes B$ denote the Kronecker product of $A$ and $B$ in $\R^{n \times n}$.
The matrix consisting of all zeros in $\R^{m \times n}$ is denoted by $0_{m \times n}$. Sometimes, we write $0_m$ to mean $0_{m\times 1}$ for the sake of compactness.
The identity matrix in $\R^{n \times n}$ is denoted by $I_n$.
We use the property that for any $C \in \R^{m \times n}, C C^T \geq 0$.
We denote the smallest and largest eigenvalues of a matrix $A$ by $\lambda_{min}(A)$ and $\lambda_{max}(A)$, respectively.
All vectors are considered as column vectors.

Consider the system evolving on $\SO{3} \times \R^3 \times \R^3 \subset \R^{3 \times 3} \times \R^3 \times \R^3$ as
\begin{subequations}\label{sys_main}
\begin{align}
    \dot{R} &= R\Omega_\times,\\
    \dot{p} &= v,\\
    \dot{v} &= g + Ra.
\end{align}
\end{subequations}
Measurements of $\Omega$ and $a$ are available with a constant additive bias in the form $\Omega_m = \Omega + b_\Omega$ and $a_m = a + b_a$, respectively.
This additive bias is a random turn-on bias, present due to thermal, physical, mechanical, and electrical properties of the sensor, and varies every time the sensors are started.
In this paper, we assume the measurements of $R$ and $p$ being available.
Usually, measurements are available as a function of $R$ and $p$.
In this case, the measurements may be constructed using available measurements.
Note that we do not assume availability of measurement of the linear velocity, $v$, of the system.

The following is assumed about the angular velocity of the system:
\begin{assum}\label{assum_ang_vel}
    The angular velocity of the system $\Omega$ is bounded.
\end{assum}

We assume measurements are available continuously.
Since we design deterministic observers, the measurements of the states are assumed to have no noise.
In practice, the measurements of the states will be noisy.
Also, the biases considered may contain noise which can be modelled as a Gauss-Markov process.
However, as numerical and experimental simulations show, the proposed observers perform satisfactorily in the presence of some noise in the measurements.

%>>>

\section{Proposed observers}%<<<
We present two observers in this section.
The first observer assumes knowledge of the bounds of $\Omega$, and provides a constant gain observer for the system.
The second observer assumes no such knowledge of the bounds, but only that a bound on $\Omega$ exists.
However, it uses a Riccati equation based observer, which is computationally expensive as compared to the first observer.

\subsection{Constant Gain Observer}%<<<
Choose the observer equations as:
\begin{subequations}\label{sys_obs_const}
    \begin{align}
        \dot{\bar{R}} &= R(\Omega_m - \bar{b}_\Omega)_\times + k_1 (R - \bar{R}),\\
        \dot{\bar{b}}_\Omega &= k_2 \Pa(R^T \bar{R})_\vee,\\
        \dot{\bar{p}} &= \bar{v} + k_3 (p - \bar{p}),\\
        \dot{\bar{v}} &= g + R (a_m - \bar{b}_a) + k_4 (p - \bar{p}),\\
        \dot{\bar{b}}_a &= -k_5 R^T (p - \bar{p}),
    \end{align}
\end{subequations}
where $(\bar{R}, \bar{b}_\Omega, \bar{p}, \bar{v}, \bar{b}_a) \in \R^{3 \times 3} \times \R^3 \times \R^3 \times \R^3 \times \R^3$, and $k_1, k_2, k_3, k_4, k_5 \in \R$.

\begin{lemma}%<<<
    Given $c > 0$, there exist $k_3, k_4, k_5 \in \R$ which satisfy the inequalities
    \begin{align}\label{eq_k3k4k5_conditions}
        \begin{aligned}
            Y := \begin{bmatrix}
                \substack{2 k_3^2 - 2 k_4\\ - k_5^2} & k_3 k_4 - k_3 k_5^2 & -k_3 k_5\\
                k_3 k_4 - k_3 k_5^2 & \substack{2 k_4^2 - 2 k_3 k_5\\ - k_3^2 k_5^2} & -k_4 k_5\\
                -k_3 k_5 & -k_4 k_5 & 2 k_5^2 - c^2
            \end{bmatrix} > 0,\quad
            Z := \begin{bmatrix}
                k_3 & k_4 & -k_5\\
                k_4 & (k_3 k_4 - k_5) & -k_3 k_5\\
                -k_5 & -k_3 k_5 & k_4 k_5
            \end{bmatrix} > 0.
        \end{aligned}
    \end{align}
\end{lemma}
\begin{proof}
    Note that $Y$ can be simplified as
    \begin{align*}
        Y =
        \begin{bmatrix}
            k_3 \\ k_4 \\ -k_5
        \end{bmatrix}
        \begin{bmatrix}
            k_3 & k_4 & -k_5
        \end{bmatrix} +
        \begin{bmatrix}
            k_5 \\ - k_3 k_5 \\ 0
        \end{bmatrix}
        \begin{bmatrix}
            k_5 & - k_3 k_5 & 0
        \end{bmatrix} +
        \begin{bmatrix}
            k_3^2 - 2 k_4 - 2 k_5^2 & 0 & 0\\
            0 & k_4^2 - 2 k_3 k_5 - 2 k_3^2 k_5^2 & 0\\
            0 & 0 & k_5^2 - c^2
        \end{bmatrix}.
    \end{align*}
    If $k_3, k_4$ and $k_5$ are chosen such that $k_5 > c$, $k_3^2 - 2 k_4 - 2 k_5^2 > 0$, and $k_4^2 - 2 k_3 k_5 - 2 k_3^2 k_5^2 > 0$, then $Y > 0$.

    The principal minors of $Z$ can be written as $k_3$, $k_3^2 k_4 - k_3 k_5 - k_4^2$, and $k_5 (k_3^2 k_4^2 - k_3^3 k_5 - k_4^3 + k_5^2)$.
    To ensure that $Z$ is positive definite, it suffices to choose $k_3, k_4$ and $k_5$ such that $k_3 > 0, k_3^2 k_4 - k_3 k_5 - k_4^2 > 0,$ and $k_3^2 k_4^2 - k_3^3 k_5 - k_4^3 > 0$ if $k_5 > 0$.

    We define the set $\mathbb{K}(c)$ containing the permissible values of $k_3, k_4$ and $k_5$ as
    \begin{align}
        \mathbb{K}(c) &= \left\{(k_3, k_4, k_5) \mid
            k_5 > c, k_3 > 0,
            k_4^2 - 2 k_3 k_5 - 2 k_3^2 k_5^2 > 0,
            k_3^2 k_4 - k_3 k_5 - k_4^2 > 0,\right.\nonumber\\
            &\left.k_3^2 k_4^2 - k_3^3 k_5 - k_4^3 > 0,
            k_3^2 - 2 k_4 - 2 k_5^2 > 0\right.\}
    \end{align}
    The set $\mathbb{K}(c)$ is an open set.
    Note that these conditions are one example of the possible conditions on $k_3, k_4$ and $k_5$ for positive definiteness of $Y$ and $Z$, and hence these conditions do not define an exhaustive set.

    If $c \leq 1$, $(k_3, k_4, k_5) = (10, 40, 2) \in \mathbb{K}(c)$.
    Since $\mathbb{K}(c)$ is an open set, there are infinitely many values in the neighbourhood of $(10, 40, 2)$ which also belong to $\mathbb{K}(c)$.
    Consider $(k_3', k_4', k_5') \in \mathbb{K}(1)$.
    Let $c > 1$, and choose $k_3 = c k_3', k_4 = c^2 k_4'$ and $k_5 = c k_5'$.
    Then,
    \begin{align*}
            k_5 = c k_5' > c,\quad
            k_3 = c k_3' > 0,\quad
            k_3^2 - 2 k_4 - 2 k_5^2 = c^2 k_3'^2 - 2 c^2 k_4' - 2 c^2 k_5'^2 > 0,\\
            k_3^2 k_4 - k_3 k_5 - k_4^2 = c^4 k_3'^2 k_4' - c^2 k_3' k_5' - c^4 k_4'^2 > c^4 k_3'^2 k_4' - c^4 k_3' k_5' - c^4 k_4'^2 > 0,\\
            k_3^2 k_4^2 - k_3^3 k_5 - k_4^3 = c^6 k_3'^2 k_4'^2 - c^4 k_3'^3 k_5' - c^6 k_4'^3 > c^6 k_3'^2 k_4'^2 - c^6 k_3'^3 k_5' - c^6 k_4'^3 > 0,\\
            k_4^2 - 2 k_3 k_5 - 2 k_3^2 k_5^2 = c^4 k_4'^2 - 2 c^2 k_3' k_5' - 2 c^4 k_3'^2 k_5'^2 > c^4 k_4'^2 - 2 c^4 k_3' k_5' - 2 c^4 k_3'^2 k_5'^2 > 0.
    \end{align*}
    Hence, $(k_3, k_4, k_5) = (c k_3', c^2 k_4', c k_5') \in \mathbb{K}(c)$, and hence satisfy inequalities \eqref{eq_k3k4k5_conditions}.
    Since $\mathbb{K}(c)$ is an open set, there exist infinite solutions in the neighbourhood of $(c k_3', c^2 k_4', c k_5')$ which satisfy inequalities \eqref{eq_k3k4k5_conditions}.
    This concludes the proof.
\end{proof}
%>>>

Define the observer error terms
\begin{subequations}\label{eq_obs_errors}
\begin{align}
    E_R = R - \bar{R}, \quad e_p &= p - \bar{p}, \quad e_v = v - \bar{v}, \\ e_\Omega = b_\Omega - \bar{b}_\Omega&, \quad e_a = b_a - \bar{b}_a.
\end{align}
\end{subequations}

\begin{theorem}%<<<
    \label{theo_const}
    Choose $k_1, k_2 > 0$ and $k_3, k_4, k_5$ satisfying inequalities \eqref{eq_k3k4k5_conditions} with $c = \sup_{t \geq 0} \| \Omega(t) \|$.
    Then the proposed observer \eqref{sys_obs_const} converges exponentially fast to the system \eqref{sys_main} under Assumption \ref{assum_ang_vel} for all $(\bar{R}(0), \bar{b}_\Omega(0), \bar{p}(0), \bar{v}(0), \bar{b}_a(0)) \in \R^{3 \times 3} \times \R^3 \times \R^3 \times \R^3 \times \R^3$.
\end{theorem}

\begin{proof}
    Differentiating the errors along the system \eqref{sys_main} and \eqref{sys_obs_const}, the error system is
    \begin{subequations}\label{sys_err_const}
    \begin{align}
        \dot{E}_R &= -R e_{\Omega_\times} - k_1 E_R,\label{sys_err_Er_const}\\
        \dot{e}_\Omega &= k_2 \Pa(R^T E_R)_\vee,\\
        \dot{e}_p &= e_v - k_3 e_p,\\
        \dot{e}_v &= -R e_a - k_4 e_p,\\
        \dot{e}_a &= k_5 R^T e_p.
    \end{align}
    \end{subequations}

    Define $x = (e_p, e_v, e_a)$.
    Then the update equation for $x$ can be written as
    \begin{align*}
        \dot{x} = \begin{bmatrix}
            -k_3 I_3 & I_3 & 0_{3 \times 3}\\
            -k_4 I_3 & 0_{3 \times 3} & -R\\
            k_5 R^T & 0_{3 \times 3} & 0_{3 \times 3}
        \end{bmatrix} x = A(t) x.
    \end{align*}

    Since the dynamics of $E_R$-$e_\Omega$ system are independent of the $x$ system, consider the Lyapunov functions
    \begin{align*}
        V_1 &= \frac{k_2}{2} \dotpr{E_R}{E_R} + \dotpr{e_\Omega}{e_\Omega},\\
        V_2 &= x^T P^{-1} x,
    \end{align*}
    where $P$ is defined as
    \begin{align*}
        P = \begin{bmatrix}
            k_3 I_3 & k_4 I_3 & -k_5 R\\
            k_4 I_3 & (k_3 k_4 - k_5) I_3 & -k_3 k_5 R\\
            -k_5 R^T & -k_3 k_5 R^T & k_4 k_5 I_3
        \end{bmatrix}.
    \end{align*}

    Differentiating $V_1$ along system \eqref{sys_err_const}, we get
    \begin{align*}
        \dot{V}_1 &= k_2 \dotpr{E_R}{-R e_{\Omega_\times} - k_1 E_R} + 2\dotpr{e_\Omega}{k_2 \Pa(R^T E_R)_\vee} = - k_1 k_2 \dotpr{E_R}{E_R},
    \end{align*}
    which is negative semi-definite.
    Hence, $E_R$ and $e_\Omega$ are bounded since $V_1$ is bounded.
    Moreover, since $V_1$ is bounded below by $0$ and non-increasing, $V_1(t)$ has a finite limit as $t \to \infty$.
    Differentiating $\dot{V}_1$ along system \eqref{sys_err_const}, we have $\ddot{V}_1 = - k_1 k_2 \dotpr{E_R}{-R e_{\Omega_\times} - k_1 E_R}$ which is bounded since $E_R$ and $e_\Omega$ are bounded, $R \in \SO{3}$ is bounded, and $\Omega$ is bounded due to Assumption \ref{assum_ang_vel}.
    Hence, $\dot{V}_1$ is uniformly continuous.
    By Barbalat's lemma, since $V_1(t)$ has a finite limit as $t \to \infty$, $\dot{V}_1(t) \to 0$ as $t \to \infty$.
    Hence, $\lim_{t \to \infty} E_R(t) = 0$.

    Differentiating equation \eqref{sys_err_Er_const}, we get
    \begin{align*}
        \ddot{E}_R &= k_1^2 E_R + k_1 R e_{\Omega_\times} - R \Omega_\times e_{\Omega_\times} - k_3 R \Pa(R^T E_R),
    \end{align*}
    which is bounded since $E_R, R$ and $e_\Omega$ are bounded, and $\Omega$ is bounded by Assumption \ref{assum_ang_vel}.
    Hence, $\dot{E}_R$ is uniformly continuous. Using Barbalat's Lemma, since $\lim_{t \to \infty} E_R(t) = 0$, $\lim_{t \to \infty} \dot{E}_R(t) = 0$.
    Hence, $e_\Omega(t) = -(R(t)^T (\dot{E}_R(t) + k_1 E_R(t)))_\vee \to 0$ as $t \to \infty$.
    Since the system under consideration is linear, the error system is exponentially stable (Theorem 4.11 \cite{Khali2002}).

    Since
    \begin{align*}
        P =
        \begin{bmatrix}
            I_6 & 0_{6 \times 3}\\ 0_{3 \times 6} & R^T
        \end{bmatrix}
        (Z \otimes I_3)
        \begin{bmatrix}
            I_6 & 0_{6 \times 3}\\ 0_{3 \times 6} & R
        \end{bmatrix},
    \end{align*}
    where $Z$ is as defined in inequalities \eqref{eq_k3k4k5_conditions}, $P$ and $Z$ have the same eigenvalues.
    Consequently, defining
    \begin{align}
        W_1(x) = \frac{\| x \|^2}{\lambda_{max}(Z)},\label{eq_W1_const}\\
        W_2(x) = \frac{\| x \|^2}{\lambda_{min}(Z)},\label{eq_W2_const}
    \end{align}
    we have that $ W_1(x) \leq V_2(t, x) \leq W_2(x) $ since $\lambda_{min}(Z) > 0$.
    Using $c = \sup_{t \geq 0}{\|\Omega(t)\|}$, define the matrices $C, D$ and $Q$ as
    \begin{align*}
        C = \begin{bmatrix} k_5 I_3 & k_3 k_5 I_3 & - R \Omega_\times \end{bmatrix},\;
        D = \begin{bmatrix}
        0_{6 \times 6} & 0_{6 \times 3} \\
        0_{3 \times 6} & c^2 I_3 - \Omega_\times^T \Omega_\times
        \end{bmatrix}, \textrm{ and }
        Q = \begin{bmatrix}
            I_6 & 0_{6 \times 3}\\ 0_{3 \times 6} & R^T
        \end{bmatrix}
        (Y \otimes I_3)
        \begin{bmatrix}
            I_6 & 0_{6 \times 3}\\ 0_{3 \times 6} & R
        \end{bmatrix},
    \end{align*}
    where $Y$ is as defined in inequalities \eqref{eq_k3k4k5_conditions}.
    Differentiating $V_2$ along system \eqref{sys_err_const},
    \begin{align*}
        \dot{V}_2 &= x^T (A(t)^T P^{-1} + P^{-1} A(t) - P^{-1} \dot{P} P^{-1}) x = x^T P^{-1} (A(t) P + P A^T(t) - \dot{P}) P^{-1} x,\\
        &= - x^T P^{-1} \left( C^T C + D + Q \right) P^{-1} x.
    \end{align*}
    Note that $C^T C \geq 0$ for all $C \in \R^{3 \times 9}$, $D \geq 0$ since $c^2 \geq \lambda_{max}(\Omega_\times^T \Omega_\times)$, and $Q > 0$ since $Y > 0$.
    Hence, $\dot{V}_2(x)$ is negative definite.

    Note that
    \begin{align*}
        P^{-1} Q P^{-1} = \begin{bmatrix}
            I_6 & 0_{6 \times 3}\\ 0_{3 \times 6} & R^T
        \end{bmatrix}
        (Z^{-1} Y Z^{-1} \otimes I_3)
        \begin{bmatrix}
            I_6 & 0_{6 \times 3}\\ 0_{3 \times 6} & R
        \end{bmatrix},
    \end{align*}
    where we have used the property $(A \otimes B) (C \otimes D) = AC \otimes BD$ and $(A \otimes B)^{-1} = A^{-1} \otimes B^{-1}$.
    Hence, the eigenvalues of $P^{-1} Q P^{-1}$ are the same as the eigenvalues of $Z^{-1} Y Z^{-1}$.
    Define
    \begin{align*}
        W_3(x) = \lambda_{min}(Z^{-1} Y Z^{-1})\| x \|^2,
    \end{align*}
    such that $W_3(x) \leq - \dot{V}_2(t, x)$, i.e. $\dot{V}_2(t, x) \leq - W_3(x)$.
    From equations \eqref{eq_W1_const} and \eqref{eq_W2_const} we have $W_1(x)$ and $W_2(x)$ such that $W_1(x) \leq V_2(t, x) \leq W_2(x)$.
    From Theorem 4.10 \cite{Khali2002}, the error system is exponentially stable.
    Hence, global exponential convergence of the constant gain observer to the state is shown.

    % Since the system under consideration is linear, asymptotic convergence implies exponential convergence of the system(Theorem 4.11 \cite{Khali2002}).
\end{proof}
%>>>

As can be seen from the required conditions, this observer requires the knowledge of the value of the bound $\sup_t\|\Omega(t)\|$.
Since this bound may not always be available, we propose the following variable gain observer.
%>>>

\subsection{Variable Gain Observer}%<<<
We first consider the following linear system
\begin{align}\label{sys_lin}
    \begin{aligned}
        \dot{x}(t) &= A(t) x(t),\\
        y(t) &= C x(t),
    \end{aligned}
\end{align}
where $A \in \R^{n \times n}$ and $C \in \R^{m \times n}$.

\begin{lemma}\label{lem_uco}
    The bounded system \eqref{sys_lin} is uniformly completely observable if there exists a $\mu > 0$ and $p \in \N$ such that for all $t$, $M(t) M^T(t) \geq \mu I_n > 0$, where $M(t)$ is defined as $M(t) = \begin{bmatrix} M_0(t) & \ldots & M_p(t) \end{bmatrix}$ with $M_0(t) = C^T(t)$ and $M_{i+1}(t) = \dot{M}_i(t) + A^T(t) M_i(t)$.
\end{lemma}

Consider the continuous Riccati equation(hereafter abbreviated as CRE)
\begin{align}\label{eq_CRE}
    \dot{P} = A(t) P + P A^T(t) - P C^T Q(t) C P + V(t),
\end{align}
where $Q \in \R^{m \times m}$ and $V \in \R^{n \times n}$ are symmetric positive definite matrices.
Let $P(t) \in \R^{n \times n}$ be the solution of equation \eqref{eq_CRE} with $P(0) \in \R^{n \times n}$ a symmetric positive definite matrix.
\begin{lemma}\label{lem_kalman}
    If system \eqref{sys_lin} is uniformly completely observable then the solution $P(t)$ is well defined on $\R^+$.
    Also, there exist constants $0 < q_m \leq q_M < \infty$ such that $q_m I_n \leq P(t) \leq q_M I_n$.
\end{lemma}
Choose the observer equations as:
\begin{subequations}\label{sys_obs_var}
    \begin{align}
        \dot{\bar{R}} &= R(\Omega_m - \bar{b}_\Omega)_\times + k_1 (R - \bar{R}),\\
        \dot{\bar{b}}_\Omega &= k_2 \Pa(R^T \bar{R})_\vee,\\
        \dot{\bar{p}} &= \bar{v} + K_3 (p - \bar{p}),\\
        \dot{\bar{v}} &= g + R (a_m - \bar{b}_a) + K_4 (p - \bar{p}),\\
        \dot{\bar{b}}_a &= K_5 (p - \bar{p}).
    \end{align}
\end{subequations}
where $(\bar{R}, \bar{b}_\Omega, \bar{p}, \bar{v}, \bar{b}_a) \in \R^{3 \times 3} \times \R^3 \times \R^3 \times \R^3 \times \R^3$, $k_1, k_2 \in \R_{> 0}$ and $K_3, K_4, K_5 \in \R^{3 \times 3}$ are given by $\begin{bmatrix}K_3^T & K_4^T & K_5^T\end{bmatrix}^T = P(t) \begin{bmatrix}I_3 & 0_{3 \times 3}  & 0_{3 \times 3}\end{bmatrix}^T Q(t)$, where $P(t) \in \R^{9 \times 9}, t \geq 0$ is the solution of the CRE equation with $P(0) = I_9$ with $Q(t) > 0$ and $A(t)$ and $C$ defined as
\begin{align}
    A(t) = \begin{bmatrix}0_{3 \times 3} & I_3 & 0_{3 \times 3}\\ 0_{3 \times 3} & 0_{3 \times 3} & -R(t) \\ 0_{3 \times 3} & 0_{3 \times 3} & 0_{3 \times 3}\end{bmatrix}, \quad C = \begin{bmatrix}I_3 & 0_{3 \times 3} & 0_{3 \times 3}\end{bmatrix}.\label{eq_AC}
\end{align}

\begin{lemma}%<<<
    There exists a solution for the CRE given in equation \eqref{eq_CRE} for the above choice of $A(t)$ and $C$.
\end{lemma}
\begin{proof}
    Consider the linear time varying system \eqref{sys_lin}.
    From Lemma \ref{lem_uco}, with $A$ and $C$ defined as in \ref{eq_AC}, we have that
    \begin{align*}
        M_0(t) = \begin{bmatrix} I_3\\ 0_{3 \times 3}\\ 0_{3 \times 3} \end{bmatrix},\; M_1(t) = \begin{bmatrix}0_{3 \times 3}\\ I_3\\ 0_{3 \times 3}\end{bmatrix},\; M_2(t) = \begin{bmatrix}0_{3 \times 3}\\ 0_{3 \times 3}\\ -R^T\end{bmatrix}.
    \end{align*}
    Hence, $M(t) M^T(t) = I_9$, implying that the system \eqref{sys_lin} is uniformly completely observable. Hence, the solution $P(t)$ of the CRE given in equation \eqref{eq_CRE} starting at a symmetric positive definite matrix exists and is well defined with the chosen $A(t)$ and $C$ from Lemma \ref{lem_kalman}.
\end{proof}
%>>>

Define the observer error terms as in equation \eqref{eq_obs_errors}.
The following theorem proposes a variable gain observer.

\begin{figure*}[!b]
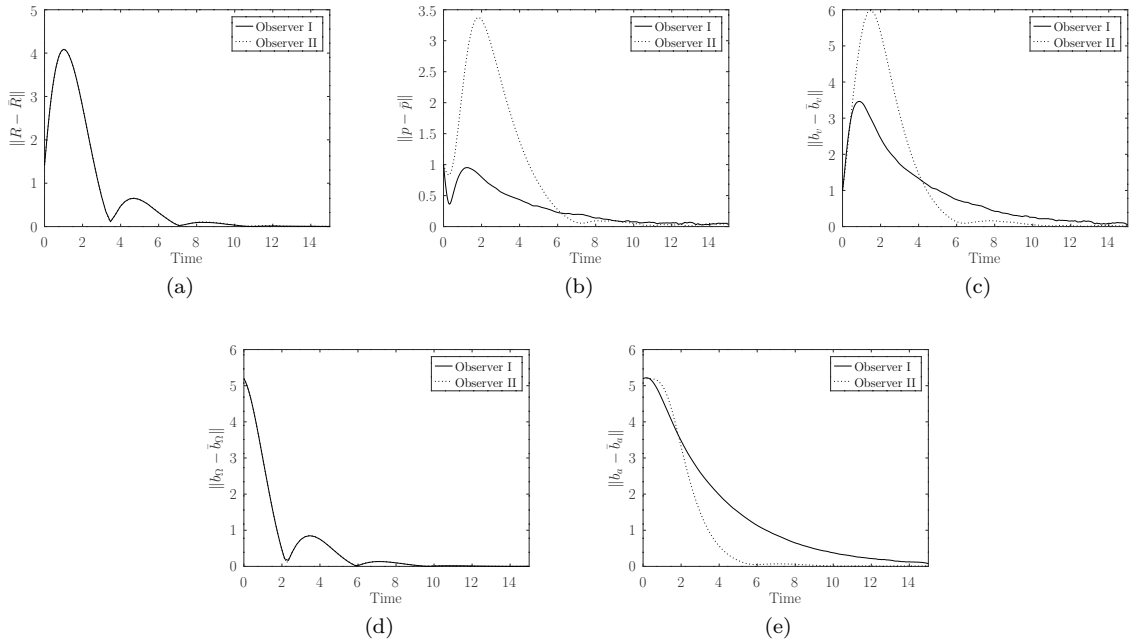

    \centering
    \subfloat[][]{\resizebox{0.3\linewidth}{!}{\input{Figures/continuous_SE3_noise_1.tex}}}\quad
    \subfloat[][]{\resizebox{0.3\linewidth}{!}{\input{Figures/continuous_SE3_noise_2.tex}}}\quad
    \subfloat[][]{\resizebox{0.3\linewidth}{!}{\input{Figures/continuous_SE3_noise_3.tex}}}\\
    \subfloat[][]{\resizebox{0.3\linewidth}{!}{\input{Figures/continuous_SE3_noise_4.tex}}}\quad
    \subfloat[][]{\resizebox{0.3\linewidth}{!}{\input{Figures/continuous_SE3_noise_5.tex}}}
    \caption{Simulation of observer in presence of noise: observer I: constant gain observer, observer II: variable gain observer}
    \label{fig_noise}
\end{figure*}

\begin{theorem}%<<<
    \label{theo_variable}
    With $Q$ and $V$ defined as constant symmetric positive definite matrices in the CRE equation \eqref{eq_CRE}, the observer system \eqref{sys_obs_var} converges to the system \eqref{sys_main} exponentially fast under Assumption \ref{assum_ang_vel} for all $(\bar{R}(0), \bar{b}_\Omega(0), \bar{p}(0), \bar{v}(0), \bar{b}_a(0)) \in \R^{3 \times 3} \times \R^3 \times \R^3 \times \R^3 \times \R^3$.
\end{theorem}

\begin{proof}
    Differentiating the errors, the error system is
    \begin{subequations}\label{sys_err_var}
    \begin{align}
        \dot{E}_R &= -R e_{\Omega_\times} - k_1 E_R,\label{sys_err_Er_var}\\
        \dot{e}_\Omega &= k_2 \Pa(R^T E_R)_\vee,\\
        \dot{e}_p &= e_v - K_3 e_p,\\
        \dot{e}_v &= -R e_a - K_4 e_p,\\
        \dot{e}_a &= -K_5 e_p.
    \end{align}
    \end{subequations}
    Define $x = (e_p, e_v, e_a)$.
    Then the update equation for $x$ can be written as
    \begin{align*}
        \dot{x} = A(t) x - P(t) C^T Q C x,
    \end{align*}
    with $A(t)$ and $C$ as defined in equation \eqref{eq_AC}.

    Since the dynamics of $R$-$\Omega$ system are independent of the $x$ system, consider the two Lyapunov functions
    \begin{align*}
        V_1 &= \frac{k_2}{2} \dotpr{E_R}{E_R} + \dotpr{e_\Omega}{e_\Omega},\\
        V_2 &= x^T P^{-1} x.
    \end{align*}
    The exponential convergence of the $E_R$-$e_\Omega$ system follows similar to that of Theorem \ref{theo_const}.

    Note that $P(t)$ is bounded above by $q_M I_9$ and below by $q_m I_9$.
    Defining
    \begin{align}
        W_1(x) &= \| x \|^2/q_M\label{eq_W1_var},\\
        W_2(x) &= \| x \|^2/q_m\label{eq_W2_var},
    \end{align}
    we have that $W_1(x) \leq V_2(t, x) \leq W_2(x)$.
    Differentiating $V_2$ along the trajectory of the system \eqref{sys_err_var},
    \begin{align*}
        \dot{V}_2 &= x^T ((A - P C^T Q C)^T P^{-1} + P^{-1} (A - P C^T Q C) - P^{-1} \dot{P} P^{-1}) x \\&= x^T (A^T P^{-1} + P^{-1} A - 2 C^T Q C - P^{-1} (A P + P A^T - P C^T Q C P + V) P^{-1}) x\\
        &= - x^T (C^T Q C + P^{-1} V P^{-1}) x
    \end{align*}
    Note that $C^T Q C \geq 0$ since $Q > 0$, and $P^{-1} V P^{-1} > 0$ since $P, V > 0$.
    Defining 
    \begin{align*}
        W_3(x) = \frac{\lambda_{min}(V) \| x \|^2}{q_M^2},% \leq \lambda_{min}(V) \| P^{-1} x \|^2,
    \end{align*}
    we see that $W_3(x) \leq \lambda_{min}(V) \| P(t)^{-1} x \|^2$ since $1/q_M$ is the smallest singular value of $P^{-1}$ due to symmetricity of $P$.
    Hence, $W_3(x) \leq - \dot{V}_2(t, x)$, i.e. $\dot{V}_2(t, x) \leq -W_3(x)$.
    Since $W_1(x) \leq V_2(t, x) \leq W_2(x)$ from equations \eqref{eq_W1_var} and \eqref{eq_W2_var}, the system is exponentially stable (Theorem 4.10 \cite{Khali2002}).
    Hence, the proposed variable gain observer is globally exponentially convergent.

\end{proof}
%>>>

%>>>

%>>>

\begin{figure*}[!t]
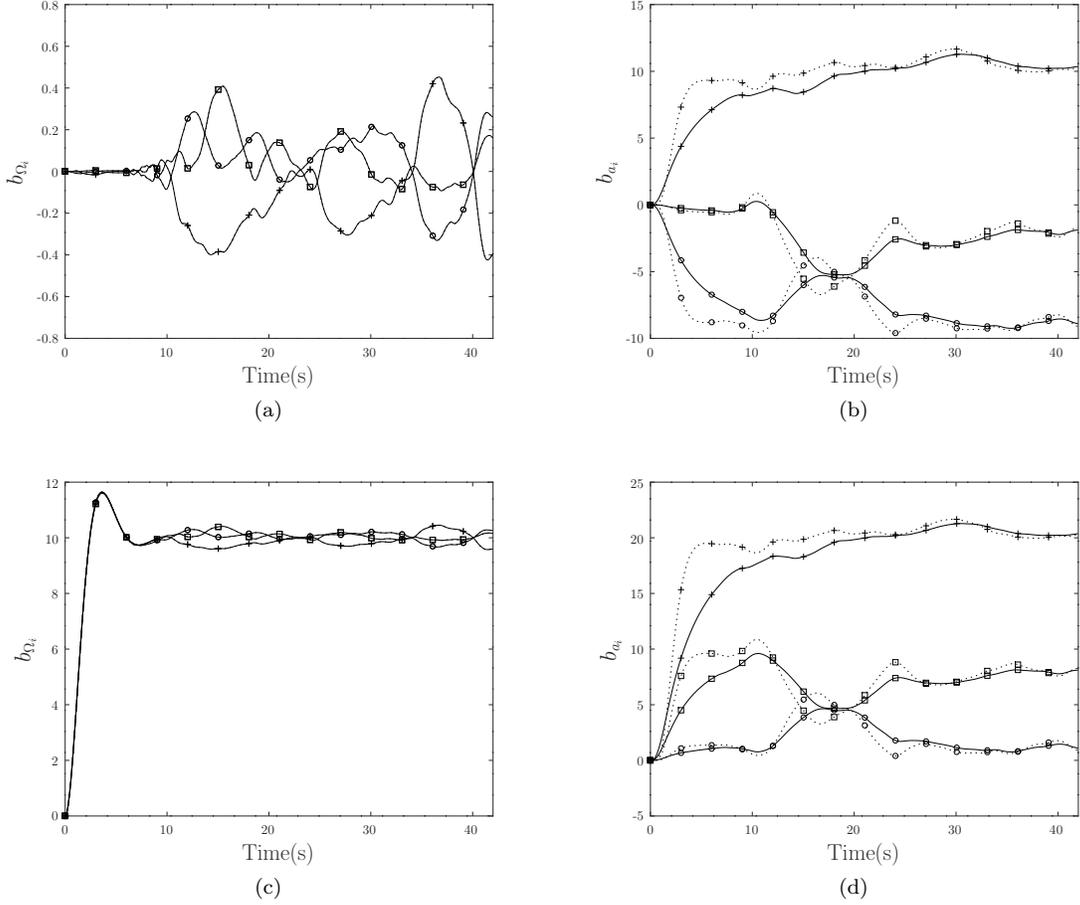

    \centering
    \subfloat[][]{\resizebox{0.45\linewidth}{!}{\input{Figures/exp_bias_1.tex}}}\quad
    \subfloat[][]{\resizebox{0.45\linewidth}{!}{\input{Figures/exp_bias_2.tex}}}\\
    \subfloat[][]{\resizebox{0.45\linewidth}{!}{\input{Figures/exp_bias_3.tex}}}\quad
    \subfloat[][]{\resizebox{0.45\linewidth}{!}{\input{Figures/exp_bias_4.tex}}}
    \caption{Experimental simulation of observer with data collected from a Intel Realsense T265: observer I: constant gain observer (solid line), observer II: variable gain observer (dashed line). Legend: square: $b_{(\Omega, a)_x}$, plus: $b_{(\Omega, a)_y}$, circle: $b_{(\Omega, a)_z}$. Figures (a) and (b) correspond to no added bias, Figures (c) and (d) correspond to added bias of 10 units.}
    \label{fig_exp}
\end{figure*}

\section{Numerical simulation}%<<<

To simulate the observers, the system is initialised as $(R(0), p(0), v(0)) = (\exp(-\pi e_{3_\times}/3), 0_{3}, 0_{3})$.
The true angular velocity and linear acceleration of the system are chosen as $(-\sin(10 t), \cos(10 t), 0.6\sin(5 t))$ and $(\cos(0.5 t), \sin(0.5 t), \cos(t))$, respectively.
The measurement of these values are corrupted with Gaussian noise of amplitude $0.01$ and constant biases $b_\Omega = (-1, 1, 5)$ and $b_a = (1, -5, 1)$, respectively.
For the measurement of the pose, we assume availability of $l$ markers (expressed as homogeneous vectors) $b = [b_1, b_2, \ldots, b_l] \in \R^{4 \times l}$ in the inertial frame and their corresponding measurements in the body frame $r = [r_1, r_2, \ldots, r_l] \in \R^{4 \times l}$, which are corrupted by Gaussian noise of amplitude $0.01$.
The measurement of the state of the system is arrived at using the landmark measurements as
\begin{align*}
    \begin{bmatrix}
        R_m & p_m \\ 0 & 1
    \end{bmatrix} = \pi_{\SE{3}}(r b^{\dagger}),
\end{align*}
where $\pi_{\SE{3}}(A)$ represents projection of $A \in \R^{4 \times 4}$ onto the special Euclidean group of three dimensions, and $b^{\dagger}$ represents the pseudo-inverse of $b$.

The observer state $(\bar{R}(0), \bar{p}(0), \bar{v}(0), \bar{b}_\Omega(0), \bar{b}_a(0)) = (I_3, 0_{3}, 0_{3}, 0_{3}, 0_{3})$ is chosen as the initial state of both the observers.
The simulation is run for 15 seconds.
We choose $k_1 = 1, k_2 = 1, k_3 = 3.4, k_4 = 5.5$ and $k_5 = 1.3$ for the constant gain observer and $k_1 = 1, k_2 = 1, P(0) = I_9, V = 0.1 I_9$ and $Q = I_3$ for the variable gain observer.
The results of the simulation is shown in Figure \ref{fig_noise}.
From the simulation, it can be seen that the observers converge to the system state fairly quickly.
Moreover, the bias value is estimated correctly even in the presence of noise in the system.
%>>>

\begin{table*}[!t]
    \centering
    \caption{Estimated Biases on data collected using Intel Realsense T265\label{tab_bias}. Observer I: constant gain observer, observer II: variable gain observer}

    \begin{tabular}{l|c|c|c|c}
        & \multicolumn{2}{c|}{Observer I} & \multicolumn{2}{c}{Observer II}\\
        \cline{2-5}
        & $b_{\Omega}$ & $b_a$ & $b_{\Omega}$ & $b_a$\\
        \hline
        No added bias & (0.05, -0.08, 0.03) & (-2.28, 9.58, -7.69) & (0.05, -0.08, 0.03) & (-2.29, 9.72, -7.76)\\
        \hline
        Added bias of 10 units & (9.90, 9.79, 9.88) & (7.26, 19.11, 1.85) & (9.9, 9.79, 9.88) & (7.29, 19.3, 1.82)\\
    \end{tabular}
\end{table*}

\section{Experiments}%<<<
We perform experimental simulations to check the feasibility of the observer to estimate the bias in the angular velocity and the linear acceleration.
We collect data of the pose using an Intel Realsense sensor.
We use an IMU for the measurement of the angular velocity of the system, and an accelerometer for the measurement of the linear acceleration.
We use the proposed observers for the estimation of the bias in the measurements.
Since the noise in the measurement devices may be comparable to the bias, we also run the observers with the same data by adding a large(as compared to the noise) constant bias.
In this experiment, we add $(10, 10, 10)$ to the measurements of the linear acceleration and angular velocity.
The convergence of the bias to a constant value ensures that the bias has been accurately estimated.

The gains are chosen as $k_1 = 1, k_2 = 1, k_3 = 3.4, k_4 = 5.5$ and $k_5 = 1.3$ for the constant gain observer and as $k_1 = 1, k_2 = 1, P(0) = I_9, V = 0.1 I_9$, and $Q = I_3$ for the variable gain observer.
The simulation results can be seen in Figure \ref{fig_exp}.
As expected, the biases converge to a constant value.
The estimated biases are shown in Table \ref{tab_bias}.
We see that the observer estimates the estimated biases well and converges to an almost constant value fairly quickly.
We also see that the two experiments arrive at a similar value of bias (after removing the added bias of 10 units), and hence are consistent.
The final estimate in the bias(taken as mean of the four biases arrived at) in the angular velocity is $(-0.025, -0.145, -0.045)$ and in the linear acceleration is $(-2.505, 9.428, -7.945)$.
%>>>

\section{Conclusion}%<<<
We have designed two observers which show globally exponentially stable convergence to the rigid body rotation and translation system under assumptions of boundedness of angular velocity.
These observers estimate the unknown bias in the linear acceleration and angular velocity measurements exponentially fast.
The first observer assumes knowledge of bounds of the angular velocity, and proposes a constant gain observer.
The second observer assumes no such knowledge, but provides a Riccati like observer, which is computationally expensive as compared to the first one.
A possible future research topic would be proposing a constant gain observer without the knowledge of the bounds of angular velocity.
An observer for discrete measurements with a continuous time system can also be designed by using the results of the proposed observer.
Also, another possible future research topic would be using the ambient space extension technique to propose a machine learning based observer for the rigid body system.
%>>>

\bibliographystyle{plain}
\bibliography{References}

\end{document}